\documentclass[11pt]{article}
\usepackage{geometry}

\usepackage[dvips]{graphics,color}
\usepackage{amsfonts}
\usepackage{amssymb}
\usepackage{amsmath}
\usepackage{etex,etoolbox}
\usepackage{environ}
\usepackage{amsthm}
\usepackage{latexsym}
\usepackage{enumitem}

\usepackage{multicol}
\usepackage[ruled, vlined]{algorithm2e}
\makeatletter
\newtheorem*{rep@theorem}{\rep@title}
\newcommand{\newreptheorem}[2]{%
\newenvironment{rep#1}[1]{%
 \def\rep@title{#2 \ref{##1}}%
 \begin{rep@theorem}}%
 {\end{rep@theorem}}}
\makeatother

\newtheorem{theorem}{Theorem}
\newreptheorem{theorem}{Theorem}

\newtheorem{lemma}{Lemma}
\newreptheorem{lemma}{Lemma}
\theoremstyle{definition}

\newtoks\magicAppendix
\magicAppendix={}
\newtoks\magictoks
\newif\iflater
\laterfalse
\long\def\later#1{\magictoks={#1}%
  \edef\magictodo{\noexpand\magicAppendix={\the\magicAppendix \par
    \noexpand\setcounter{theorem}{\arabic{theorem}}%
    \the\magictoks}}%
  \magictodo}
\long\def\both#1{\later{#1}\the\magictoks}
\def\magicappendix{\latertrue \the\magicAppendix}

\newbox\gnBoxA
\newdimen\gnCornerHgt
\setbox\gnBoxA=\hbox{$\ulcorner$}
\global\gnCornerHgt=\ht\gnBoxA
\newdimen\gnArgHgt
\def\q #1{%
\setbox\gnBoxA=\hbox{$#1$}%
\gnArgHgt=\ht\gnBoxA%
\ifnum     \gnArgHgt<\gnCornerHgt \gnArgHgt=0pt%
\else \advance \gnArgHgt by -\gnCornerHgt%
\fi \raise\gnArgHgt\hbox{$\ulcorner$} \box\gnBoxA %
\raise\gnArgHgt\hbox{$\urcorner$}}

\newcommand{\argmax}{\operatornamewithlimits{argmax}}

\newcommand{\logdet}[1]{\log \det \of{#1}}

\renewcommand{\O}[1]{\mathcal{O}\of{#1}}

\newcommand{\of}[1]{\left(#1\right)}
\renewcommand{\b}[1]{\left\{#1\right\}}
\newcommand{\logo}[1]{\log\of{#1}}

\newcommand{\sumsub}[2]{\sum_{\substack{#1 \\ #2}}}

\renewcommand{\P}[1]{\mathbb{P}\of{#1}}

\newcommand{\E}[1]{\mathbb{E}\left[#1\right]}
\newcommand{\Ec}[2]{\mathbb{E}\left[#1\middle\vert #2\right]}

\newcommand{\abs}[1]{\left\vert#1\right\vert}

\renewcommand{\epsilon}{\varepsilon}

\newcommand{\users}{[N]}

\newcommand{\REGRET}{\O{\rho \sqrt[3]{NT^2}}}

\newcommand{\PtH}{p^{\leq T}\of{H}}
\newcommand{\OPTtH}{\text{OPT}^{\leq T}\of{H}}
\newcommand{\OPTH}{\OPTtH}
\newcommand{\xot}{x_0^t}
\newcommand{\xit}{x_1^t}
\newcommand{\pt}{p^t}
\newcommand{\st}{s^t}

\begin{document}
\title{Provably Manipulation-Resistant Reputation Systems}
\author{ Paul Christiano\thanks{UC Berkeley. \ Email: paulfchristiano@eecs.berkeley.edu}}
\date{}
\maketitle

\begin{abstract}

We consider a community of users
who must make periodic decisions about whether
to interact with one another.
We propose a protocol which allows honest users to 
reliably interact with each other,
while limiting the damage done by each
malicious or incompetent user.
The worst-case cost per user is sublinear in the average number
of interactions per user and is independent of the number of users.
Our guarantee holds simultaneously for every group of honest
users. For example, multiple groups of users with incompatible
tastes or preferences can coexist.

As a motivating example, we consider a game where
players have periodic opportunities to do one another favors
but minimal ability to determine when a favor was done.
In this setting, our protocol achieves nearly optimal collective welfare
while remaining resistant to exploitation.

Our results also apply to a collaborative filtering
setting where users must make periodic decisions about
whether to interact with resources such as movies or restaurants.
In this setting, we guarantee that any set of honest users
achieves a payoff nearly as good as if they had identified
the optimal set of items in advance
and then chosen to interact only with resources
from that set.

\end{abstract}

\section{Introduction}\label{intro}

Trust and reciprocity
distinguish repeated interactions
within a small village from isolated interactions
in a massive online community.
In a small group we are willing to trust each other and expect
that ``what goes around comes around,''
but online we are more guarded:
we rely on trusted mediators,
expect others to do the bare minimum they have promised,
and rarely leave ourselves vulnerable to exploitation.
Can the advantages of trust within a small group
be realized at a global scale, in an environment
where repeat interactions are rare
and manipulation will be rampant if it is profitable?

We help address this question by showing
that a subset of individuals who intend to trust and help each other
can reap the benefits even in
the presence of a large majority of exploitative manipulators.
This result continues to hold when
feedback is extremely noisy,
individuals engage in only a handful of interactions,
and the manipulators coordinate adversarially.
Our per-user performance guarantees are completely independent of the size of the 
community.

%Many social systems are dependent on notions of \emph{trust}, \emph{reputation}, and \emph{reciprocity}.
%Although many protocols have been proposed for tracking
%reputation in large systems on the internet,
%these systems still seem to perform poorly
%compared to word of mouth and reciprocity norms
%in conventional social networks.
%Moreover, the performance of conventional social
%behaviors at a large scale is itself an open question.
%These familiar protocols work well in small communities, 
%when deception is difficult and serious
%attempts at manipulation are relatively rare;
%but there may be fundamental difficulties to scaling these ideas
%to a larger, more sophisticated, and more competitive world.
%
We use online learning over positive semidefinite matrices to iteratively adjust
a symmetric measure of trust.
This symmetric measure is coupled with a non-trivial system
for tracking reciprocal obligations.
The result is a protocol which guarantees that the honest users perform
nearly as well as if they had simply excluded all dishonest users.

\subsection{Motivating example}\label{example}

Consider a set of $N$ users with periodic opportunities to do one another \emph{favors}.
In each round, two users Alice and Bob are chosen at random;
Alice may choose to a pay a cost of $1$,
and if she does then Bob receives an expected benefit of $1 + \epsilon$.

If the users were maximizing collective welfare
then they would always take these opportunities.
After $kN$ rounds the average user's payoff would be $\epsilon k$.
But such an altruistic community is open to exploitation by defectors,
who receive favors but don't reliably pay them back.
Is there an altruistic strategy which can achieve high collective welfare
but isn't vulnerable to such exploitation?

We assume that a favor confers an \emph{expected} benefit of $1 + \epsilon$,
but that the actual benefit is either $0$ or $\rho$.
In this case three distinct outcomes are indistinguishable to most users:
\begin{itemize}
\item Alice does Bob a favor, but by chance the benefit is $0$.
\item Alice claims to have done a favor for Bob, though in fact she did not.
\item Bob claims not to have received a favor, though in fact he did.
\end{itemize}

Suppose some set of $\alpha N$ ``honest'' users is willing to follow whatever protocol
we specify.
If they knew each other's identities in advance, then after $kN$
rounds these users could achieve an average benefit
of $\alpha \epsilon k$ 
by always doing favors for each other and never doing unreciprocated favors for outsiders.
We show how to obtain a result nearly this good without knowing the identities of the honest users in advance.
Namely, we show that there is a protocol which allows the average honest user a payoff
of $\alpha \epsilon k  - \O{ \frac {\rho k^{2/3}}{\alpha} }$ after $k N$ rounds,
\emph{regardless} of how the dishonest users behave or how many of them there are.

\section{The model}

Consider a set of $N$ users
and a sequence of rounds, $t = 0, 1, \ldots$.
In each round, a pair of users $x_{0}^t, x_{1}^t \in \users$ is chosen by nature,
along with a pair of payoffs $p_0^t, p_1^t \in [-\rho, \rho]$.
All users learn the values $\xot, \xit$.
The users $x_i^t$ choose whether to interact, picking $s_i^t \in \b{0, 1}$.
If either user picks $0$ then they don't interact
and both receive a payoff of $0$.
If both users pick $1$ then they do interact and $x_i$ receives a payoff $p_i^t$.
Each user learns the payoff they receive
but not what payoff their partner received.

For the analysis, we fix a set $H \subset \users$ of ``honest'' users
who will follow whatever protocol we prescribe and whose welfare we will seek to maximize.
Our guarantees will simultaneously hold for any set $H$
whose members follow our protocol.
We will assume that the users have access to unlimited cheap talk
(in the form of broadcasts),
verifiable identities,
and a shared random beacon.\footnote{The assumptions 
of verifiable identities, broadcast communication, and a random beacon
can be replaced by the use of signature schemes and coin flipping protocols
or by the availability of a trusted central party.}

\newcommand{\Ptx}{p^{\leq T}\of{x}}
The \emph{payoff} of a user $\Ptx$ is the sum of their payoffs in rounds
they participated in:
\[ \Ptx = \sumsub{i, t \leq T: x_i^t = x}{s_0^t = s_1^t = 1} p_i^t.\]
Let $\PtH$ be the total payoff of the honest users after $T$ rounds, i.e.
\[ \PtH = \sum_{x \in H} \Ptx. \]
We define the benchmark $\OPTtH$ to be the counterfactual payoff of the honest users
if they had always interacted with each other and only with each other:
\begin{align*}
\OPTtH = \sumsub{t \leq T}{x_{0}^t, x_{1}^t \in H } \of{ p_0^t + p_1^t }.
\end{align*}
In some of our results we include the possibility of \emph{transfers}.
If transfers are available, player $x_i^t$ chooses a payment $\tau_i^t \geq 0$
to make to their partner at the end of round $t$.
Player $x_0^t$ then receives an additional payoff of $\tau_1^t - \tau_0^t$
and player $x_1^t$ receives an additional payoff of $\tau_0^t - \tau_1^t$.
These additional payoffs are included in the sum defining $\Ptx$.

Our framework only considers interactions between users,
but by introducing users who always agree to interact
and always receive a payoff of $0$
we can easily model static \emph{resources}
such as movies, news articles, or merchants who don't care
whom they interact with.
Some of our results concern the \emph{collaborative filtering} setting,
in which each interaction involves one resource and one user.
For notational convenience, we will assume that the resource is always $x_1^t$
and the user is always $x_0^t$.
In this setting the set $H$ contains both honest users and some set of resources.
To obtain the best bounds we will usually imagine that $H$ contains the ``optimal''
set of items according to the users in $H$, i.e. all items which yield
a positive payoff on average.

We have described the interactions as being proposed ``by nature''
but these interactions can also be proposed by the users.
Our protocol can be used as a black box within a larger system
designed to maximize $\OPTH$.\footnote{This could involve 
using machine learning or simple heuristics 
to predict which interactions will be positive, 
or identifying items
a user wants to buy and proposing a popular merchant who sells them,
or any number of alternatives. 

It also falls on the system to estimate
the opportunity cost of a mediocre interaction, as discussed in Section~\ref{results}.}

\section{The results}\label{results}

Our main result is the following:
\begin{theorem}\label{main-result}
There is a randomized polynomial time protocol
which yields expected performance
\[ \PtH \geq \OPTtH - \REGRET \]
for every set $H$ of honest users,
assuming that one of the following conditions is satisfied:
\begin{enumerate}
\item Transfers are available.
\item In every round $p_1^t = 0$. (The collaborative filtering setting.)
\item In every round with $\xot, \xit \in H$
the payoffs are symmetric \emph{ex ante}:
$\E{p_0^t} = \E{p_1^t}$,
conditioned on $x_0^t, x_1^t$, and all information available before round $t$.
\end{enumerate}
\end{theorem}
Note that the favor setting described in Section~\ref{example}
satisfies this \emph{ex ante} symmetry condition,
provided that Alice has an opportunity to do Bob a favor
with the same probability that Bob has an opportunity to do Alice a favor.\footnote{Some
condition of this form is clearly necessary. For example,
suppose that Alice has an opportunity to do Bob a favor in \emph{every} round.
If Alice does not take these opportunities, then
we have failed to maximize collective welfare.
But if Alice does take these opportunities, then we have failed to maximize Alice's welfare.
Since we need to satisfy the theorem for both $H = \b{\textrm{Alice}}$
and $H = \b{\textrm{Alice}, \textrm{Bob}}$, either way we are in trouble.}

If we hold fixed the ratio $k = T/N$,
then the regret is $\O{\rho N k^{2/3}}$.
So the introduction of each dishonest user increases the total regret by $\O{\rho k^{2/3}}$
over the first $k$ interactions.\footnote{
The addition of an honest user also increases the regret by $\O{\rho k^{2/3}}$,
but typically decreases the \emph{average} regret per honest user.
If an $\alpha$ fraction of users are honest,
then the average regret per honest user is $\O{\frac {\rho k^{2/3}}{\alpha}}$.}

If we also hold fixed the ratio $\alpha = \abs{H}/N = \alpha$,
then the per-user regret bound is \emph{independent} of $N$.
In massive communities, this requires two interacting users to make informed
decisions even if they are connected only
by a long and tenuous chain of mutual acquaintances.

We also prove the following lower bound on the achievable regret:
\begin{reptheorem}{lower-bound}
For every protocol and every $T$,
there exists a set of payoffs and a set $H$
such that
\[ \PtH \leq \OPTtH - \Omega\of{ \rho \sqrt{N T} }\]
in expectation,
even assuming that $p_0^t = p_1^t$ in every interaction.
\end{reptheorem}
That is, over the first $k$ interactions it is necessary
for the honest users to collectively
pay a cost of $\Omega\of{\rho k^{1/2}}$ per dishonest user
just to get enough information to statistically identify them.

There is a tension between choosing a smaller set $H$---for which
the optimum payoff may be higher---and choosing a larger
set $H$ in order to amortize the loss $\REGRET$ over a larger set.
For collaborative filtering in particular,
recommendations optimized for a smaller group of users will be strictly better.
This tension is only in the analysis, however,
since the guarantee applies simultaneously
for every group $H$.
As $T$ grows and the term $\REGRET$ becomes less important,
the best analysis might obtain for smaller and smaller groups $H$,
and to compete the algorithm must make more and more finely-tuned recommendations.

We note that our model differs from many other approaches to reputation systems
or collaborative recommendation;
our protocol merely \emph{filters} proposed interactions, rather than recommending
optimal interactions. It is up to the users of the system to propose interactions,
and also to estimate the opportunity cost of mediocre interactions
(we assume that any interaction with a positive payoff was good,
so $0$ needs to represent the ``indifferent'' outcome; watching a good movie
should have a negative payoff, if I would otherwise watch a great movie).

In some settings this seems like an acceptable model. For example,
I might browse an online marketplace, find an item that looks like a good deal,
and then use the merchant's reputation to decide whether to purchase it.
Or I might see a movie which looks interesting
and use reviews to decide whether to watch it.

In other settings this is a very inappropriate model. For example,
I might have the choice between hundreds of merchants offering the same item.
In this case I just want to be recommended the best one.
Our results would need to be strengthened to be applicable in such settings.

See Section~\ref{further-work} 
for a more extensive discussion of modeling issues, limitations, and directions
for further work.
Our protocol makes a qualitatively novel guarantee, but there is ample
room for improvement.
For example, our algorithm is not distributed,
the protocol is not incentive-compatible for honest users,
and we only consider sequential interactions.

%In order to use this protocol in the favor setting,
%the honest users ignore information about whether they are the potential
%benefactor of beneficiary. If both parties agree to interact
%and an honest user is the benefactor,
%then they always do a favor if possible.
%Of course the dishonest users may behave arbitrarily.
%Depending on the model, 
%dishonest users might do a favor for an honest user even if the honest user
%doesn't agree to interact;
%of course this will only make the payoff of the honest users better.
\section{The protocol}

Our main algorithmic contribution is an online learning algorithm
which could be used by a central planner to decide which pairs of users should interact.
The primary difficulty is that the set $H$ is not known in advance
and so we need to simultaneously lower bound the total payoff
of \emph{every} set $H$.
This is an unusual guarantee, and requires a novel technical approach.
This algorithm is outlined in Section~\ref{planneroverview}
and formally analyzed in Section~\ref{centralplanner}.

This algorithm can then be used directly to achieve our guarantees
when $p_1^t = 0$ or when $p_0^t = p_1^t$.
If transfers are available, we reduce to the case $p_0^t = p_1^t$
by using transfers to equalize the two payoffs.
These reductions are routine,
and are analyzed in Section~\ref{simplereductions}.

The most challenging condition in Theorem~\ref{main-result} is \emph{ex ante} symmetry.
To address this case we introduce a digital currency, and use transfers
of this currency to equalize payoffs.
The primary difficulty is ensuring that there is sufficient
liquidity to facilitate transactions between honest users.
We propose a robust protocol based on logarithmic utility functions.
This system is outlined in Section~\ref{currencyoverview}
and formally analyzed in Section~\ref{exantereduction}.

\subsection{Outline of the central planner's algorithm}\label{planneroverview}

Consider the following single-payoff problem faced by a central planner:
in each round a pair of individuals $\xot, \xit \in \users$ is given,
the planner outputs $\st \in [0, 1]$ representing the probability
with which $\xot$ and $\xit$ should interact,
and then nature reveals a payoff $\pt$ for the interaction.

\newcommand{\PtHi}{p^{\leq T}_i\of{H}}
\newcommand{\OPTHi}{\OPTH}
There are two different payoffs associated with a set $H$ in this game,
depending on whether we think of the payoff as accruing to player $0$
or player $1$:
\[ \PtHi = \sum_{\substack{t \leq T:x_i^t \in H}} \st \pt. \]
We can create a symmetric version, in which the payoff
is shared equally between the two players,
by averaging $p^{\leq T}_0$ and $p^{\leq T}_1$.
All of these different payoff notions are associated with the same benchmark:
\[ \OPTHi = \sum_{\substack{t \leq T : \xot, \xit \in H}} \pt. \]

\newcommand{\plannersalg}{\mathrm{PLAN}}
We exhibit an algorithm $\plannersalg$ satisfying:
\begin{reptheorem}{thm:centralplanner}
Suppose that $p^t s^t \in [-\rho, \rho]$
and that $\E{p^t} \in [-\rho, \rho]$ conditioned on $x_0^t, x_1^t$ and the information
from all previous rounds.
Then for every $H, T$ and $i$, $\plannersalg$ achieves payoffs satisfying
\begin{align*}
\PtHi &\geq \OPTHi - \REGRET \\
\end{align*}
in expectation.
\end{reptheorem}
This protocol can be used almost directly to achieve
low regret when $p_1^t = 0$ (by applying the lower bound on $p_0^{\leq T}\of{H}$)
or when $p_0^t = p_1^t$ (by applying the lower bound on $p^{\leq T}\of{H} = p_0^{\leq T}\of{H} + p_1^{\leq T}\of{H}$).

The basic outline of our algorithm  $\plannersalg$ is:
\begin{itemize}
\item The planner initializes a matrix $P^t_{ab}$, indexed by $a, b \in \users$,
with $P^0_{ab} = 1$.
\item In step $t$, the planner outputs $s^t = P^t_{\xot\xit}$.
\item After step $t$, the planner uses an online learning
algorithm to predict matrices $C^{t+1}_{ab}$ (for \emph{cut}) and $J^{t+1}_{ab}$
(for \emph{join}).
%Intuitively, a high value of $C^t_{ab}$ indicates the recommendation
%that users $a$ and $b$ should stop interacting,
%while a high value of $J^t_{ab}$ indicates the recommendation
%that users $a$ and $b$ should interact more.
\item The planner then updates 
\begin{align*}
P^{t+1}_{ab} &= P^t_{ab} - C^{t+1}_{ab} P^t_{ab} + J^{t+1}_{ab} (1 - P^t_{ab}) \\
&= (1 - J^{t+1}_{ab} - C^{t+1}_{ab}) P^t_{ab} + J^{t+1}_{ab}
\end{align*}
The matrices $C^{t+1}_{ab}$ and $J^{t+1}_{ab}$ are chosen by the online learning algorithm
with the goal of maximizing $p^{t+1} s^{t+1}$.
\end{itemize}
$J^t$ and $C^t$ are predicted by the online local learning algorithm
from \cite{me}, which achieves payoffs as high
as the best fixed matrices of the form
\begin{align*}
J_{ab} &= \P{\ell(a) = \ell(b) = 0} \\
C_{ab} &= \P{\of{\ell(a) = 0,  \ell(b) = 1} \vee \of{\ell(a) = -1, \ell(b) = 0}}
\end{align*}
for any distribution over labelings $\ell : \users \rightarrow \b{-1, 0, 1}$.

The most counterintuitive feature of the algorithm is that the predicted
matrices are used to adjust $P$, rather than to directly
determine which users should interact.
This allows the protocol to build up more complex patterns
of interaction as time goes on,
and is needed to minimize regret
for \emph{every} set $H$ simultaneously.
Other approaches we tried, such as using an online learning
algorithm with a richer comparison class,
were unable to achieve this guarantee.

The analysis proceeds in several steps:
\begin{itemize}
\item The online learning guarantee implies that the matrices $P^{t}_{ab}$
perform nearly as well as any sequence of the form
\[ Q^{t}_{ab}(J^*, C^*) = (1 - J^*_{ab} - C^*_{ab}) P^{t-1}_{ab} + J^*_{ab}, \]
where $J^*$ and $R^*$ are constant matrices from the comparison class.
That is, the recommendations $J^t$ and $C^t$ are nearly optimal
within this class.
\item The strong smoothness lemma for the online learning algorithm
implies that $J^t_{ab}$ and $C^t_{ab}$ change slowly.
Using this fact and a non-trivial analysis of our update rule
we prove that $\abs{P^t_{ab} - P^{t-1}_{ab}}$
is small on average.
This implies that the matrices $Q^{t}$ perform nearly as well as
the matrices
\[ R^{t}_{ab}(J^*, C^*) = (1 - J^*_{ab} - C^*_{ab}) P^{t}_{ab} + J^*_{ab}. \]
(The superscript on $P$ has changed from $t-1$ to $t$;
this difference may appear small, but it leads to a much harder
analysis and ultimately to our suboptimal dependence on $T$.)
\item We show that if $p^{\leq T}_0\of{H} \ll \OPTH$, then
there is some pair of matrices $J^*, C^*$ in the comparison
class such that the matrices $R^{t}(J^*, C^*)$ substantially outperform
the matrices $P^{t}$. 
Namely, we can take $J^*_{ab} = 1$ for $a, b \in H$,
and $C^*_{ab} = 1$ for $a \in H, b \not \in H$.
Similarly,  
if $p^{\leq T}_1\of{H} \ll \OPTH$ then
we can take $C^*_{ab} = 1$ for $a \not \in H, b \in H$
\item (1) and (2) together imply that the matrices $R^{t}(J^*, C^*)$ cannot much outperform $P^{t}$
for any $J^*, C^*$ in the comparison class.
Combined with (3), this implies that we can't have $\PtHi \ll \OPTH$.
\end{itemize}

Surprisingly, the bottleneck
in the quantitative analysis is actually in step (2);\footnote{though if this 
bottleneck were addressed,
the rest of our analysis would need to be tightened in order
to actually improve the results}
we take a small learning rate
to control the differences $\abs{P^{t+1}_{ab} - P^t_{ab}}$.
This is not just a technical requirement.
Because the central planner's updates are both optimal and small,
we can infer that there are no big adjustments to be made.
This property is what lets us ensure 
that the central planner is making adjustments often enough
to ``exhaust'' all of the good ones.

\subsection{Protocol for \emph{ex ante} symmetric interactions}\label{currencyoverview}

Recall the \emph{ex ante} symmetry condition:
\[ \Ec{p_0^t - p_1^t}{x_0^t, x_1^t, K^{<t}} = 0 \]
where $K^{<t}$ is a transcript of all of the events prior to round $t$.
In this section we describe a protocol which achieves
low regret under this condition.

Our approach is to assign each user a level of wealth $w(x) > 0$.
A high value of $w(x)$ indicates a debt owed to a user by the rest of the community,
while a low value indicates a debt owed \emph{by} the user.
Whenever $p_0^t \neq p_1^t$, we adjust $w\of{\xot}$ and $w\of{\xit}$
so that the ``effective payoffs'' of the two users are equal
and the total wealth $w\of{\xot} + w\of{\xit}$ is conserved.
Because the effective payoffs are symmetric,
we can apply the central planner's algorithm to guarantee
that the total effective payoff is large.
Over the long-run the budget constraints $w\of{x} > 0$
imply that the actual payoffs can't be too different
from the effective payoffs.

We first considered various straightforward approaches
in which payoffs were directly translated into changes in wealth,
i.e. which equalized the effective payoffs
\[ p_0^t + \Delta w\of{\xot} = p_1^t + \Delta w\of{\xit}. \]
The difficulty with these proposals was maintaining the invariant $w(x) > 0$.
In order to prevent a violation, we need to prevent interactions with users who have low $w(x)$.
But this means that if insufficient liquidity is available, interactions between
honest users will stop.
This issue made it very difficult to analyze protocols of this type,
and in most cases were able to find concrete situations in which
these protocols performed poorly,\footnote{Manipulators can break
many systems by hoarding currency, either by doing too many nice things or by slightly
understating their own payoffs. This creates a liquidity problem
amongst the honest users, which is difficult to correct without creating enough currency
for the dishonest users to run up destructively large deficits. 
If this attempted manipulation fails, 
manipulators can slightly understate their own payoffs when interacting with some honest users
and slightly overstate them when interacting with others, creating an imbalance
of payments among the honest users which has a similar effect.}
even in the simple favor setting described in Section~\ref{example}.

Instead we define the effective payoff
as the actual payoff plus the change in \emph{log} wealth,
and equalize this quantity:
\[ p_0^t + \Delta \logo{w\of{\xot}} = p_1^t + \Delta \logo{w\of{\xit}}. \]
This automatically leads our learning protocol to place less weight on the payoffs of users
who run up substantial deficits, preventing them from causing too much damage.
At the same time, we never have to block an interaction between honest users due to insufficient funds;
as $w(x) \rightarrow 0$ the system places a higher and higher value on changes in $w(x)$,
which prevent it from ever setting $w(x) \leq 0$.
The resulting protocol seems to be self-correcting
and highly robust to liquidity shortages or imbalances of payments.

\section{Related work}\label{related}

In this section we discuss a number of related threads in the literature.
Our novel contribution is a performance guarantee
in the face of general attacks and in a general environment,
rather than considering narrower threat models
or making restrictive assumptions about the domain.

\textbf{Collaborative filtering:} There is a large literature on collaborative filtering
spanning a wide range of techniques; see \cite{survey} for an overview.

Technically the most closely related results
are matrix prediction algorithms
that achieve low regret compared to the comparison class of matrices
with small trace norm,
and in particular the recent results for online matrix prediction
due to Hazan, Kale, and Shalev-Shwartz \cite{matrix-prediction}
and the similar work by the author \cite{me}.
However, filtering systems
based directly on these algorithms would be vulnerable to manipulation.

Another line of related work
investigates the recovery of preference functions;
for example see \cite{tellme, strangers, recommenders}.
Some of this work uses assumptions similar to our implicit assumptions, 
for example the existence of large groups with concordant preferences.
This line of work differs in that it aims to reconstruct the entire preference table
rather than simply filtering proposed interactions.
This stronger goal is unattainable in our noisy and adversarial setting.

\textbf{Competitive collaborative learning:} From a modeling perspective, our work is most 
similar to, and inspired by,
the competitive collaborative learning framework of Awerbuch and Kleinberg \cite{ccl}.
In this setting each user must choose a resource in a series of rounds,
and the aim is to ensure that the honest players perform as well as if they
had all chosen the single best fixed resource in every round.
Like us, they rely on machinery from online learning
and convert online learning guarantees into regret bounds
for the honest users.

However, this work focuses on identifying an optimal resource to use in every round,
rather than considering binary decisions across varying conditions.
As a result, we use quite different technical tools and are applicable
in very different settings.

\textbf{Transitive trust:} Another literature explicitly examines chains of trust
relationships, by which one user can indirectly trust another;
for example see \cite{trust1, trust2}.
The techniques in this literature tend to focus on a binary notion of links
as being either honest or dishonest.
This assumption is appropriate in some settings 
(for example with respect to cryptographic credentials)
but not in our target applications.

\textbf{Manipulation-resistance:} Another literature addresses the question of modifying recommendation systems
to limit the influence of sybils, fake identities controlled by an attacker. 
For example, see \cite{dsybil, influence-limiter}.
The robustness guarantees achieved in this literature are similar to our own,
but the results are restricted to the setting where a decision-maker
can get useful feedback from the same source across many different decisions.
In our setting
only a handful of users have information about each decision,
and so a decision-maker can only rarely
get direct advice from the same source more than once.
In exchange, we achieve weaker bounds on the damage done per sybil.

\textbf{Economic systems for P2P networks:} Many decentralized economic
protocols have been proposed to reduce the scope for free-riding
and enforce reciprocity, especially in the context of Peer-to-Peer (P2P) 
networks;
for some representative examples see
\cite{bittorrent, bartercast, karma, minaxi, distributedwork, sybilproofaccounting}.

Driven by our desire to obtain end-to-end performance guarantees,
we focus on robustly addressing liquidity constraints and possible payment
imbalances, issues which have received little theoretical attention.
These issues are particularly important because we assume that users
only sometimes have the opportunity to provide valuable services,
rather than assuming that such opportunities are always readily available.

Conversely, we do not confront any of the practical issues involved in tracking payments.
Our protocol
would need to be implemented on top of another system for managing digital payments.
We also rely on our symmetric reputation system to provide appropriate incentives
for honest reporting,
sidestepping another prominent issue in the existing literature.

\textbf{Other reputation systems:} Reputation systems have been the focus of substantial research
across a wide range of contexts,
for example see \cite{trust-survey, eigentrust, peertrust}.
P2P networks and online marketplaces have enjoyed the most attention.

To make a decision in these systems, a user
consults feedback from peers who have made similar decisions,
and trusts each source of feedback according to their past reliability.
This is conceptually similar to our framework,
although we use different technical tools and the resulting protocol
is much more resistant to manipulation.

It may be possible to bridge the gap between
our proposal and more practical reputation systems.
For example, EigenTrust \cite{eigentrust} is a reputation system which uses
random walks to compute the maximum eigenvector of an interaction graph.
Although it is vulnerable to manipulation, the protocol is much
more efficient than ours.
By changing the online learning algorithm used by the central planner
in our protocol,
it might be possible to replace our use of semidefinite programming
by a distributed algorithm based on random walks, achieving the best of both approaches.

\newpage
\bibliographystyle{plain}

\bibliography{reputation}

\section{Appendix: The central planner's algorithm}\label{centralplanner}

\subsection{The algorithm}

\newcommand{\myalg}{\mathcal{OLL}_{\epsilon}}
\newcommand{\myalgr}{\mathcal{OLL}_{\epsilon / \rho}}
\newcommand{\myalgo}[1]{\myalg\of{#1}}
\newcommand{\myalgro}[1]{\myalgr\of{#1}}
\newcommand{\mprod}[2]{\left\langle #1, #2 \right\rangle}

\newcommand{\Xspace}{\mathbb{R}^{3N \times 3N}}
We use the algorithm $\myalg$ for online local learning from \cite{me}.
We will apply that algorithm to labelings $\ell : \users \rightarrow \b{-1, 0, 1}$,
i.e. to $3N \times 3N$ matrices with rows and columns indexed by elements of 
$\users \times \b{-1, 0, 1}$.
We consider a sequence of matrices $E^t \in \Xspace$.
Write $E^{\leq T} = \sum_{t \leq T} E^t$ and similarly for $E^{< T}$.
Then $\myalgo{E^{< T}}$
is defined as the solution to the following semidefinite program:\footnote{The 
algorithm in \cite{me} maximized under some additional equality constraints,
e.g. $\sum_i X_{(a, i)(a, i)} = 1$.
For our purposes it doesn't matter whether we include these or not,
and for ease of exposition we omit them.
It is easy to verify that this leaves the results of \cite{me} unchanged;
using these extra inequalities would also leave our results unchanged.}
\[ \argmax_{ \substack{ X \in \Xspace \\
X \succeq 0 \\
0 \leq X_{(a, i)(b, j)} \leq 1}} \epsilon \mprod{X}{E^{<T}} + \logdet{X+I}\]
where $\mprod{X}{E}$ is the matrix inner product $\sum_{a,i,b,j} X_{(a,i)(b,j)}E_{(a,i)(b,j)}$.

For us, the important properties of $\myalg$ are given by the following results:
\begin{reptheorem}{oll}\cite{me}
For any sequence $E^t$ and any $X \succeq 0$ with entries in $[0, 1]$,
we have
\[ \sum_{t < T} \mprod{E^t}{\myalgo{E^{<t}}} \geq \sum_{t<T} \mprod{E^t}{X} - \O{\epsilon \sum_{t < T} \abs{E^t}_1^2 } - \O{N \epsilon^{-1}} \]
where $\abs{E^t}_1$ is the sum of the absolute values of the entries of $E^t$.
\end{reptheorem}
\begin{replemma}{smoothness}\cite{me}
For any $E^{< t}$ and any $E^t$ we have
\[ \abs{\myalgo{E^{<t}} - \myalgo{E^{\leq t}}}_{\infty} \leq \O{\epsilon \abs{E^t}_1} \]
\end{replemma}
Our statement of Theorem~\ref{oll} and Lemma~\ref{smoothness} are slightly different 
from the statements in \cite{me},
so we present proofs in Section~\ref{ollsection}. 

\newcommand{\plannersalgeps}{\plannersalg_{\epsilon}}
We can now define the central planner's algorithm $\plannersalgeps$.
This algorithm is parametrized by a learning rate $\epsilon < \frac 1{16}$;
in order to achieve the unparametrized version we decrease $\epsilon$ over time.
Let be the sparse matrix with $Z_{\xot\xit}^t = p^t$.
Given $\b{Z^t}_{t < T}$, the following algorithm gives
the central planner's recommendations in round $T$:
\begin{itemize}
\item Set $P^0_{ab} = 1$ for $a, b \in \users$.
\item For $t = 0, 1,\ldots, T-1$:
\begin{itemize}
\item Define $E^t$ to be the sparse matrix with the following non-zero entries:
\begin{align*}
E^t_{(a, 0)(b, 0)} &= (1 - P^t_{ab}) Z^t_{ab} \\
E^t_{(a, 0)(b, 1)} &= - P^t_{ab} Z^t_{ab} \\
E^t_{(a,-1)(b, 0)} &= - P^t_{ab} Z^t_{ab}.
\end{align*}
\item Let $X^{t+1} = \myalgro{E^{\leq t}}$. 
\item Define
\begin{align*}
C^{t+1}_{ab} &= \frac 14 X^{t+1}_{(a,0)(b,1)} + \frac 14 X^{t+1}_{(a,-1)(b,0)} \\
J^{t+1}_{ab} &= \frac 14 X^{t+1}_{(a,0)(b,0)} + \epsilon^{1/2}
\end{align*}
\item Update $P$ by $P^{t+1}_{ab} =  (1 - C^{t+1}_{ab} - J^{t+1}_{ab}) P^{t+1}_{ab} + J^{t+1}_{ab}$.
The construction of $E^t$ ensures that 
$p^t P^{t}_{\xot\xit}$ is a linear function
of $\mprod{X^t}{E^t}$.
\end{itemize}
\item Output $s^T = P^T_{x_0^Tx_1^T}$.
\end{itemize}

\subsection{Analysis}
For convenience, we will write $p^t\of{P}$ for the payoff the planner
would obtain by using the matrix $P$ in round $t$, namely
\[ p^t\of{P} = p^t P_{x_0^t x_1^t}. \]

\newcommand{\update}[2]{U\of{#1,#2}}
We write $\update{P^t}{X}$ for the matrix that would be obtained by updating $P^t$
with the matrix $X^{t+1} = X$,
i.e. defined by the equations
\begin{align*}
C_{ab} &= \frac 14 X_{(a,0)(b,1)}+\frac 14 X_{(a,-1)(b,0)}, \\
J_{ab} &= \frac 14 X_{(a,0)(b,0)} + \epsilon^{1/2}, \\
\update{P^t}{X}_{ab} &= (1 - C_{ab} - J_{ab}) P^t_{ab} + J_{ab}
\end{align*}

\begin{lemma}\label{performance}
For any $X \succeq 0$ with $0 \leq X_{(i, a)(j, b)} \leq 1$
we have
\[ \sum_{t \leq T} p^t\of{P^t} \geq \sum_{t \leq T} p^t\of{\update{P^{t-1}}{X}} -
\O{\frac {\epsilon}{\rho} \sum \abs{p^t}^2 } - 
\O{\frac {\rho N}{\epsilon}},
\]
where $P^t$ are the matrices computed by $\plannersalgeps$.
\end{lemma}
\begin{proof}
We have $p^t\of{P^t} = p^t\of{\update{P^{t-1}}{X^t}}$.
Moreover,
\begin{align*}
p^t\of{\update{P^{t-1}}{X}}
&= p^t \of{ (1 - J_{\xot\xit} - C_{\xot\xit})P^{t-1}_{\xot\xit} + J_{\xot\xit }} \\
&= p^t\of{P^{t-1}_{\xot\xit}} + J_{\xot\xit} p^t \of{1 - P^{t-1}_{\xot\xit}} - C_{\xot\xit}p^t P^{t-1}_{\xot\xit} \\
&= p^t\of{P^{t-1}} + \frac 14 \mprod{X}{E^t} + \beta^t,
\end{align*}
where $J$ and $C$ are defined from $X$ in the same way that $J^t$ and $C^t$ are defined from $X^t$,
and $\beta^t = \epsilon^{1/2} p^t (1 - P^{t-1}_{\xot\xot})$
doesn't depend on $X$.
We constructed the matrix $E^t$ precisely so that it would satisfy this identity.

Because $X^t = \myalgro{E^{<t}}$, by Theorem~\ref{oll} we have
\begin{align*}
\sum_{t \leq T} p^t\of{P^t} 
&= \sum_{t \leq T} \of{p^t\of{P^{t-1}} + \frac 14 \mprod{X^t}{E^t} + \beta^t} \\
&\geq \sum_{t \leq T} \of{p^t\of{P^{t-1}} + \frac 14 \mprod{X}{E^t} + \beta^t} -
\O{\frac {\epsilon}{\rho} \sum \abs{E^t}_1^2 } - 
\O{\frac {\rho N}{\epsilon}} \\
&= \sum_{t \leq T} p^t\of{\update{P^{t-1}}{X}} -
\O{\frac {\epsilon}{\rho} \sum \abs{E^t}_1^2 } - 
\O{\frac {\rho N}{\epsilon}} \\
&= \sum_{t \leq T} p^t\of{\update{P^{t-1}}{X}} -
\O{\frac {\epsilon}{\rho} \sum \abs{p^t}^2 } - 
\O{\frac {\rho N}{\epsilon}},
\end{align*}
where the last equality holds because $\abs{E^t}_1 = 3\abs{p^t}$. \end{proof}

\begin{lemma}\label{slowchange}
For any $a, b$, we have
\[ \sum_{t < T} \abs{P_{ab}^{t+1} - P_{ab}^t} = \O{\frac{\epsilon^{1/2}}{\rho} \sum_{t <T} \abs{p^t}} + \O{1},\]
where $P^t$ are the matrices computed by $\plannersalgeps$.
\end{lemma}
\begin{proof}
Consider the ``steady state'' value of $P$ given $J$ and $C$, defined as
\[ Y^t_{ab} = \frac {J^t_{ab}}{J^t_{ab} + C^t_{ab}}. \]
This is the value to which $P$ would converge if $J$ and $C$ were left unchanged
for many iterations.
We will show that $Y^t$ does not change very quickly,
and that $P^t$ is ``following'' $Y^t$ and hence does not change very quickly either.

Note that each entry of $J^t$ and $C^t$ is defined to be a constant multiple
of an entry of $\myalgro{E^{<t}}$,
and that the entry-wise $\ell_1$ norm of $E^t$ is $\O{p^t}$.
So by Lemma~\ref{smoothness}, 
we know that $\abs{J^{t+1}_{ab} - J^t_{ab}} \leq \frac{\epsilon}{\rho} \abs{p^t}$,
and similarly $\abs{C^{t+1}_{ab} - C^t_{ab}} \leq \frac{\epsilon}{\rho} \abs{p^t}$.
By definition we have $C^{t}_{ab} \geq 0$ and $J^t_{ab} \geq \epsilon^{1/2}$.
Together these facts imply that
$\abs{Y^{t+1}_{ab} - Y^t_{ab}} = \O{\frac {\epsilon^{1/2}}{\rho} \abs{p^t}}$.

By the definition of $P^{t+1}$, we have
\begin{align*}
P^{t+1}_{ab} &= (1 - J^t_{ab} - C^t_{ab}) P^t_{ab} + J^t_{ab} \\
&= (1 - J^t_{ab} - C^t_{ab}) P^t_{ab} + (J^t_{ab} + C^t_{ab}) Y^t_{ab} \\
\end{align*}
By the definition of $J$ and $C$,
we have $J^t_{ab} + C^t_{ab} \leq \frac 14 + \frac 14 + \frac 14 + \epsilon^{1/2} \leq 1$.
Thus $P^{t+1}$ is a convex combination of $P^t_{ab}$ and $Y^t_{ab}$.
In particular, 
\begin{align*}
\abs{P^{t}_{ab} - Y^t_{ab}} = \abs{P^{t}_{ab} - P^{t+1}_{ab}} + \abs{P^{t+1}_{ab} - Y^t_{ab}}.
\end{align*}

Intuitively this suggests that $\sum{\abs{P^{t+1}_{ab} - P^t_{ab}}} \leq \sum{\abs{Y^{t+1}_{ab} - Y^t_{ab}}}$,
from which the desired result would follow.

To see this formally, consider the total length of the path
\[ P^0_{ab}, P^1_{ab}, \ldots, P^t_{ab}, Y^t_{ab}, \]
i.e. the sum of the absolute values of the differences between consecutive points.
If we change this path by adding a new point:
\[ P^0_{ab}, P^1_{ab}, \ldots, P^t_{ab}, P^{t+1}_{ab}, Y^t_{ab}\]
we don't change the total length, because $P^{t+1}$ is in between $P^t$ and $Y^{t}$.
If we then change this path by changing the last point:
\[ P^0_{ab}, P^1_{ab}, \ldots, P^t_{ab}, P^{t+1}_{ab}, Y^{t+1}_{ab} \]
we increase the total length by at most $\abs{Y^{t+1}_{ab} - Y^t_{ab}}$.

By induction, the total length of the path is thus $\sum{\abs{Y^{t+1}_{ab} - Y^t_{ab}}} + \O{1}$,
and this total length upper bounds $\sum{\abs{P^{t+1}_{ab} - P^t_{ab}}}$.
\end{proof}

The next step is to prove that \emph{if} the honest players are not collectively
performing well, \emph{then} there is some recommendation the central planner could have made
which would generate a substantial improvement.

The reason for the $\epsilon^{1/2}$ loss in this theorem is our restriction that $J^t \geq \epsilon^{1/2}$,
which lower bounds the central planner's ability to stop honest and dishonest users from interacting
and was needed to prevent rapid changes in $P^t$.

\begin{lemma}\label{improvementconversion}
For any set $H$ there is a matrix $X^H \succeq 0$,
with $0 \leq X^H_{(i, a)(j, b)} \leq 1$, such that
\[ \OPTHi - \PtHi =
\O{ \sum_{t \leq T} \of{p^t\of{\update{P^{t}}{X^H}} - p^t\of{P^t} + \epsilon^{1/2} \abs{p^t}} },\]
where $P^t$ are the matrices produced by $\plannersalgeps$.
\end{lemma}
\begin{proof}
Observe that $\plannersalgeps$ is perfectly symmetrical under an exchange of $\xot$ and $\xit$:
the only difference is that the matrices $P^t$ are transposed,
and the indices of $X_{(a, i)(b, j)}$ are modified by exchanging $+1$ and $-1$.
So it suffices to prove the theorem for $i = 0$.

Define the vector $v^H$ via $v_{(a, 0)} = 1$ for $a \in H$,
$v_{(a, 1)} = 1$ for $a \not\in H$,
and $v_{(a, i)} = 0$ otherwise.
Define the matrix $X^H$ as $\of{v^H} \of{v^H}^T$.
This matrix is clearly PSD, and has entries in $[0, 1]$.

Suppose $a \not \in H$. Then $X^H_{(a, i)(b, j)} = 0$,
and so $\update{P^t}{X^H}_{ab} = P^t_{ab}$.
Alternatively, suppose $a \in H$.
Then $X_{(a, -1)(b, 0)} = 0$, and exactly one of
$X_{(a, 0)(b, 0)}$ and $X_{(a, 0)(b, 1)}$ is $1$,
according to whether $b \in H$.

So for $a, b \in H$, we have
\[ \update{P^t}{X^H}_{ab} = \frac 34 P^t_{ab} + \frac 14 + \epsilon^{1/2} (1 - P^t_{ab}). \]
For $a\in H, b \not \in H$, we have
\[ \update{P^t}{X^H}_{ab} = \frac 34 P^t_{ab} + \epsilon^{1/2} (1 - P^t_{ab}). \]

We compute:
\begin{align*}
\sum_{t \leq T : \xot \in H} \of{\update{P^t}{X^H}} &\geq
\frac 34 \sum_{t \leq T: \xot \in H} p^t\of{P^t} + \frac 14 \sum_{t \leq T: \xot, \xit \in H}p^t - \epsilon^{1/2} \sum{\abs{p^t}} \\
&= \PtHi - \frac 14 \PtHi + \frac 14 \OPTHi - \epsilon^{1/2} \sum{\abs{p^t}}
\end{align*}
Rearranging, and using the fact that $\update{P^t}{X^H}_{ab} = P^t_{ab}$ for $a \not\in H$:
\begin{align*}
\OPTHi - \PtHi 
&= 4 \sum_{t \leq T : \xot \in H} \of{ p^t\of{\update{P^t}{X^H}} - p^t\of{P^t} + \epsilon^{1/2} \abs{p^t}} \\
&= 4 \sum_{t \leq T} \of{ p^t\of{\update{P^t}{X^H}} - p^t\of{P^t} + \epsilon^{1/2} \abs{p^t}},
\end{align*}
as desired.
\end{proof}

We are now ready to prove the main result of this section:
\begin{theorem}\label{undoubledplanner}
Suppose that $p^t s^t \in [-\rho, \rho]$ and $\E{p^t} \in [-\rho, \rho]$ conditioned
on the preceding rounds.
Then
\[ \PtHi \geq \OPTHi - \O{\rho \epsilon^{1/2} T}  - \O{\frac {\rho N}{\epsilon}} \]
in expectation, where $\PtHi$ is the payoff achieved by $\plannersalgeps$.
\end{theorem}
\begin{proof}
By Lemma~\ref{improvementconversion}, there exists an $X^H \succeq 0$ with $0 \leq X^H_{(a, i)(b, j)} \leq 1$
such that
\begin{align*}
\OPTHi - \PtHi 
&= \O{ \sum_{t \leq T} \of{p^t\of{\update{P^{t}}{X^H}} - p^t\of{P^t} + \epsilon^{1/2} \abs{p^t}} } \\
&\leq
\O{ \sum_{t \leq T} \of{p^t\of{\update{P^{t}}{X^H}} - p^t\of{P^t}} + \rho \epsilon^{1/2} T }
\end{align*}
in expectation, where the second equality holds because $\E{\abs{p^t}} \leq \rho$.
So it suffices to bound $\sum_{t \leq T} \of{p^t\of{\update{P^{t}}{X^H}} - p^t\of{P^t}}$.

Note that $\plannersalgeps$ always satisfies $s^t = P^t_{\xot\xit} \geq J^t_{\xot\xit} \geq \epsilon^{1/2}$,
and consequently $\abs{p^t}\leq \frac {s^t}{\rho} = \frac {\epsilon^{1/2}}{\rho}$.

We have
\begin{align*}
&\sum_{t \leq T} \of{p^t\of{\update{P^{t}}{X^H}} - p^t\of{P^t}} \\
&= \sum_{t \leq T} \of{p^t\of{\update{P^{t}}{X^H}} - p^t\of{\update{P^{t-1}}{X^H}}}
+ \sum_{t \leq T} \of{p^t\of{\update{P^{t-1}}{X^H}} - p^t\of{P^t}} \\
&\leq \sum_{t \leq T} p^t\abs{\update{P^{t}}{X^H} - \update{P^{t-1}}{X^H}}_{\infty}
+ \O{\frac {\epsilon}{\rho} \sum \abs{p^t}^2 } 
+ \O{\frac {\rho N}{\epsilon}} \\
&\leq \sum_{t \leq T} p^t\abs{P^t - P^{t-1}}_{\infty}
+ \O{\epsilon^{1/2}\sum \abs{p^t} }
+ \O{\frac {\rho N}{\epsilon}} \\
&\leq \rho \sum_{t \leq T} \abs{P^t - P^{t-1}}_{\infty}
+ \O{\epsilon^{1/2}\sum \abs{p^t} } 
+ \O{\frac {\rho N}{\epsilon}} \\
&\leq \epsilon^{1/2} \sum_{t \leq T} \abs{p^t}
+ \O{\epsilon^{1/2}\sum \abs{p^t} }
+ \O{\frac {\rho N}{\epsilon}} \\
&= \O{\rho \epsilon^{1/2} T} + \O{\frac {\rho N}{\epsilon}} 
\end{align*}
where the third line uses 
the definition of $p^t$
and
Theorem~\ref{performance},
the fourth line uses the definition of $\update{\cdot}{\cdot}$
and the upper bound $\abs{p^t} \leq \rho \epsilon^{-1/2}$,
the fifth line uses the bound $\E{p^t} \leq \rho$ conditioned on previous rounds
(and in particular conditioned on $\abs{P^t - P^{t-1}}_{\infty}$,
which can be computed before round $t$)
the sixth line uses Theorem~\ref{slowchange},
and the seventh line uses the bound $\E{p^t} \leq \rho$.

\end{proof}

By setting $\epsilon = \frac {N^{2/3}}{T^{2/3}}$
and using a standard doubling trick, we obtain a more convenient form:
\begin{theorem}\label{thm:centralplanner}
Under the same conditions as Theorem~\ref{undoubledplanner},
there is an algorithm $\plannersalg$ for the central planner
which achieves an expected payoff of
\[ \PtHi \geq \OPTH - \REGRET.\]
\end{theorem}
\begin{proof}
Every time $T = 2^k$, we start a new copy of $\plannersalgeps$
with $\epsilon = \epsilon_k = N^{2/3} T^{-2/3}$.

By Theorem~\ref{undoubledplanner},
we have
\begin{align*}
\sumsub{2^k \leq t < 2^{k+1}}{x_it \in H} p^t s^t 
&\geq \sumsub{2^k \leq t < 2^{k+1}}{\xot, \xit \in H} p^t 
- \O{\rho \epsilon_k^{1/2} 2^k }
- \O{\frac {\rho N}{\epsilon_k}} \\
&= \sumsub{2^k \leq t < 2^{k+1}}{\xot, \xit \in H} p^t 
- \O{\rho 2^{2/3 k} N^{1/3}}
\end{align*}

Summing from $k = 0$ to $k = \left\lceil \logo{T} \right\rceil$,
we obtain the desired bound.
\end{proof}

\section{Appendix: Algorithms for symmetric payoffs, collaborative filtering, and transfers}\label{simplereductions}

%In this section we consider distributed algorithms for symmetric payoffs,
%collaborative filtering,
%and the general problem with transfers.
%All of these algorithms work in basically the same way.
%Each user simulates the central planner independently
%and determines the central planner's recommendations.
%A pair of users then agree on the probability with which they should interact,
%and use the random beacon or a coin-flipping protocol to decide whether
%they should interact.
%If they interacted they then report the payoff fo the central planner,
%after dividing by the probability of interaction in order to produce
%an unbiased estimator.

\subsection{Symmetric payoffs}

We say that a game has symmetric payoffs if, in each interaction
between honest users,
$p_0^t = p_1^t$.

\newcommand{\symmetricalg}{\text{SYM}}
The algorithm $\symmetricalg$ followed by an honest user is:
\begin{itemize}
\item Initialize a copy of $\plannersalg$.
\item For $t = 0, 1, \ldots$:
\begin{itemize}
\item If you are interacting:
\begin{itemize}
\item Consult $\plannersalg\of{\xot, \xit}$
to find the probability $s^t$ with which you should interact.
\item Using the random beacon
decide whether you should interact,
and broadcast the result.
If you and your partner disagree about whether you should interact,
don't interact.
\item If you decide to interact, observe your payoff and broadcast it.
Otherwise broadcast a payoff of $0$.
\end{itemize}
\item Report a payoff of $-\rho$
to your instance of $\plannersalg$
if anything goes wrong:\footnote{In fact outputting a payoff of $-\infty$
would be most natural, since this should never happen in an interaction
between honest users, but this would require modifying our analysis.}
if the interacting parties broadcast different payoffs,
if either party makes multiple
broadcasts, if either party broadcasts
a payoff outside of $[-\rho, \rho]$, 
or if you disagree about whether an interaction should have occurred
(based on your independent use of the random beacon).
\item Otherwise, report a payoff of $p^t = p_0^t / s^t = p_1^t / s^t$
to your instance of $\plannersalg$.
\end{itemize}
\end{itemize}

We can now apply Theorem~\ref{thm:centralplanner}
\begin{theorem}
If $p_0^t = p_1^t$ in any interaction between honest users,
then the algorithm $\symmetricalg$ achieves a payoff
\[ \PtH \geq \OPTH - \REGRET \]
\end{theorem}
\begin{proof}
It is easy to verify by induction that all of the honest users' instances
of $\plannersalg$ will be perfectly synchronized,
and consequently that nothing will ever go wrong in an interaction
between honest users.
This also allows us to talk unambiguously about ``the''
instance of $\plannersalg$.

Let $p^t$ be the payoff reported to the central planner
in round $t$.
Either $p^t = 0$, $p^t = -\rho$, or $p^t = p_0^t / s^t$ and $p_0^t \in [-\rho, \rho]$.
Thus $p^t s^t \in [-\rho, \rho]$; because $p^t = p_0^t / s^t$ with probability at most $s^t$,
and otherwise $p^t \in [-\rho, \rho]$, we have
$\E{p^t}  \in [-2\rho, 2\rho]$.
Thus we can apply Theorem~\ref{thm:centralplanner} to the payoffs $p^t$
(the difference between $\rho$ and $2\rho$ makes no difference in the asymptotics).

In any interaction between honest users, we have $\E{p^t} = p_0^t = p_1^t$.
Moreover, in any interaction involving an honest user,
either their expected payoff is equal to $s^t \E{p^t}$,
or else they agreed to interact and $p^t = -\rho$.
In either case, their expected payoff is at least $s^t \E{p^t}$.

Thus we can apply Theorem~\ref{thm:centralplanner} twice and conclude
\begin{align*}
\sumsub{i, t : x_i^t \in H}{s_0^t = s_1^t = 1} p_i^t
&\geq \sum_{t :\xot \in H} s^t p^t + \sum_{t : \xit \in H} s^t p^t  \\
&\geq 2 \sum_{t : \xot, \xit \in H} p^t - \REGRET\\
&= \sum_{t : \xot, \xit \in H} \of{p_0^t + p_1^t} - \REGRET\\
&= \OPTH - \REGRET
\end{align*}
in expectation.
\end{proof}

\subsection{Collaborative filtering}

\newcommand{\filteringalg}{\text{FILTER}}

Our algorithm $\filteringalg$ for collaborative filtering is a nearly direct application of $\plannersalg$.
Formally, it works as follows.

\begin{itemize}
\item Initialize a copy of $\plannersalg$.
\item For $t = 0, 1, \ldots$:
\begin{itemize}
\item If you are $\xit$, output $s_1^t = 1$.
\item If you are $\xot$:
\begin{itemize}
\item Consult $\plannersalg\of{\xot, \xit}$
to find the probability $s^t$ with which you should interact.
\item Using the random beacon, output $s_0^t = 1$ with probability $s^t$
and $s_0^t = 0$ otherwise. Broadcast the result.
\item If you interact, broadcast your payoff $p_0^t$. Otherwise
broadcast $0$.
\end{itemize}
\item Report a payoff of $-\rho$
to your instance of $\plannersalg$
if $\xot$ broadcasts a payoff outside of $[-\rho, \rho]$
or if you disagree about whether an interaction should have occurred
(based on your independent use of the random beacon).
\item Otherwise, report a payoff of $p^t = p_0^t / s^t$
to your instance of $\plannersalg$.
\end{itemize}
\end{itemize}

\begin{theorem}
If $p_1^t = 0$ in any interaction between honest users,
then the algorithm $\filteringalg$ achieves a payoff
\[ \PtH \geq \OPTH - \REGRET \]
\end{theorem}
\begin{proof}
It is easy to check that the honest user's instantiations of $\plannersalg$
remain synchronized, that $\E{p^t} \in [-\rho, \rho]$, and that $p^t s^t \in [-\rho, \rho]$.
In an interaction where $\xot$ is honest,
it is easy to see that
$\E{s^t p^t} = \E{s^t_0 s_1^t p^t_0}$.
In an interaction where $\xot, \xit$ are both honest,
then $s_1^t = 1$ and we have
$\E{p^t} = p^t_0$.

So we have in expectation
\begin{align*}
\sumsub{i, t \leq T : x_i^t \in H }{s_0^t = s_1^t = 1} p_i^t
&= \sum_{t \leq T : x_0^t \in H} s^t p^t \\
&\geq \sum_{t \leq T : \xot, \xit \in H} p^t - \REGRET\\
&= \sum_{t \leq T: \xot \xit \in H} p_0^t - \REGRET\\
&= \OPTH - \REGRET,
\end{align*}
as desired.
\end{proof}

\subsection{Transfers}\label{transferreduction}

\newcommand{\tinstance}{\mathcal{I}_{\text{transfers}}}
\newcommand{\sinstance}{\mathcal{I}_{\text{symmetric}}}
If the payoffs aren't symmetric, but users are allowed to make transfers,
then we can reduce to the symmetric case.
That is, given an instance $\tinstance$ of the general problem with transfers, with payoffs $p^t_i$,
we will produce a new instance $\sinstance$ of the symmetric problem, with payoffs $q^t_i$.
We will guarantee that if the uses perform well on this synthetic instance $\sinstance$,
then they also performed well on $\tinstance$.
Since we know that $\symmetricalg$ performs well on $\sinstance$,
this will yield the desired result.

Whenever an honest user $x_i^t$ has to make a decision in $\tinstance$
they ask the user $x_i^t$ to make the same decision in $\sinstance$.
Whenever an honest user is given a payoff $p_i^t$ in $\tinstance$,
they perform the following process in order to generate the payoff $q_i^t$
in $\sinstance$:
\begin{itemize}
\item Send the message $p_i^t$ to $x_{1-i}^t$,
and receive the message $p_{1-i}^t$.
\item If no message is received, or multiple messages are received,
or the message contains a payoff $p_{1-i}^t \not \in [-\rho, \rho]$,
then set $\tau_i^t = 0$.
\item Otherwise set $\tau_i = \max\b{\frac {p_{1-i}^t - p_i^t}{2}, 0}$.
\item Report $q_i^t = p_i^t + \tau_{1-i} - \tau_i$.
\end{itemize}
It is easy to verify that in any interaction between honest users we have
\[ q_0^t = q_1^t = \frac {p_0^t + p_1^t}{2}.\]
This allows us to apply the algorithm $\symmetricalg$ to the instance $\sinstance$.
\begin{theorem}
There is an algorithm attaining the payoff
\[ \PtH \geq \OPTH - \REGRET \]
in general if transfers are available.
\end{theorem}
\begin{proof}
We apply the reduction described in this section.
In any interaction between honest users we have 
and in any interaction we have $q_i^t \in [-\rho, \rho]$.
So we can apply $\symmetricalg$ to the resulting payoffs,
and we have in expectation
\begin{align*}
\PtH
&=\sumsub{i, t : x_i^t \in H}{s_0^t = s_1^t = 1} \of{p_i^t  + \tau_{1-i}^t - \tau_i^t} \\
&= \sumsub{i, t : x_i^t \in H}{s_0^t = s_1^t = 1} q_i^t \\
&\geq \sum_{t : \xot, \xit \in H} \of{q_0^t + q_1^t} -\REGRET\\
&= \sum_{t : \xot, \xit \in H} \of{p_0^t + p_1^t} - \REGRET,\\
&= \OPTH - \REGRET
\end{align*}
as desired.
\end{proof}

\section{Appendix: Ex ante symmetry}\label{exantereduction}

\newcommand{\exsinstance}{\mathcal{I}_{\text{ex ante}}}
As in Section~\ref{transferreduction} 
we reduce an instance $\exsinstance$ with ex ante symmetric payoffs
to an instance $\sinstance$ with symmetric payoffs.
At the beginning of the protocol
each user initializes a balance $w^0\of{x} = 1$
for every user $x$.
When $x_i^t$ receives the payoff $p_i^t$ in $\exsinstance$ they do the following
to generate the payoff $q_i^t$ in $\sinstance$:
\begin{itemize}
\item Communicate the payoff $p_i^t$ to $x_{1-i}^t$,
and hear the payoff $p_{1-i}^t$.
\item If no message is heard, or multiple messages are heard,
or the message contains a payoff $p_{1-i}^t \not \in [-\rho, \rho]$,
then set $\tau_i = 0$
\item Otherwise
$\tau_i^t = \frac {\delta w_0^t w_1^t \of{p_{1-i}^t - p_i^t}}{2 \rho \of{w_0^t + w_1^t}}$,
where $w_i^t = w^t\of{x_i^t}$:
\item Each user updates $w^{t+1}\of{x_i^t} = w^t\of{x_i^t} + \tau^t_i$.
For $x \not\in \b{\xot, \xit}$, $w^{t+1}\of{x} = w^t\of{x}$.
\item Report the payoff $q_i^t = p_i^t + \frac {2 \rho}{\delta} \frac {\tau^t_i}{w^t_i}$.
\end{itemize}
Here $\delta < \frac 12$ is a parameter of the reduction.

It is easy to verify that in any interaction between honest users,
both users will report a payoff of 
\[ q_0^t = q_1^t = \frac {w_0 p_0^t + w_1 p_1^t}{w_0 + w_1}. \]
This is very similar to the protocol with transfers,
except that the payoff of an interaction is now a weighted average of the individual payoffs.
Once again, we can apply the algorithm $\symmetricalg$ to $\sinstance$.
The aim the analysis will be to show that if the payoff in $\sinstance$
is good, then so is the payoff in $\exsinstance$,
which is less straightforward than in the case with transfers.

It is easy to verify that $\abs{\tau_i^t} \leq \delta w_i^t$,
so that each $w^t\of{x}$ changes by at most a factor of $(1 \pm \delta)$ in any round.
Since $\delta < \frac 12$, this guarantees that $w^t\of{x} > 0$.

The mechanism used to track the quantities $w^t\of{x}$
is not particularly important;
for simplicity, and because our algorithm is already
hopelessly undistributed, we assumed that every user independently
tracks the quantities $w^t\of{x}$ for every other user.
It would be straightforward, however, to use a more efficient implementation
of digital cash such as \cite{creditnetworks}.\footnote{Though this introduces
the possibility of attacks on the payment infrastructure itself.} In 
this case a dishonest player might decline to make an appropriate transfer,
but this can be handled in the same way as in Section~\ref{transferreduction}.

To analyze this protocol we rely on the potential functions
\[ U^T(x) = \frac {2\rho}{\delta} \logo{w^T(x)} + p^{< T}(x) \]
We use three lemmas to relate the reported payoffs $q_i^t$ to the actual
payoffs $p_i^t$.

Our first lemma relates the reported payoffs $q_i^t$ to the underlying payoffs $p_i^t$:
\begin{lemma}\label{exantereduction1}
Under the ex ante symmetry condition,
in any interaction between honest users we have
\[\E{q_0^t + q_1^t} = \E{p_0^t + p_1^t}.\]
\end{lemma}
\begin{proof}
Under the ex ante symmetry condition, conditioned on the values of $w^t_0$ and $w^t_1$,
and on the users' decision to interact,
we have $\E{p_0^t} = \E{p_1^t}$.
Let $\beta_0 = \frac {w^t_0}{w^t_0 + w^t_1}$ and $\beta_1 = \frac {w^t_1}{w^t_0 + w^t_1}$.
Then we have
\begin{align*}
\E{q_0^t + q_1^t}
&= 2 \E{\beta_0 p_0^t + \beta_1 p_1^t} \\
&= 2 \beta_0 \E{p_0^t} + 2 \beta_1 \E{p_1^t} \\
&= \beta_0\E{p_0^t + p_1^t} + \beta_1 \E{p_0^t + p_1^t} \\
&= \E{p_0^t + p_1^t}
\end{align*}
as desired.
\end{proof}

\newcommand{\xiit}{x_i^t}
Our second lemma relates the reported payoffs $q_i^t$ to the change in $U^t\of{x}$:
\begin{lemma}\label{exantereduction2}
If $s_0^t = s_1^t = 1$, then
then $U^{t+1}\of{x_i^t} \geq U^t\of{x_i^t} + q_i^t - \O{\rho \delta}$.
\end{lemma}
\begin{proof}
We have
\begin{align*}
U^{t+1}\of{x_i^t} 
&= \frac {2\rho}{\delta} \logo{w^{t+1}\of{\xiit}} + p^{<t+1}\of{\xiit}\\
&= U^t\of{x_i^t} + \frac {2\rho}{\delta} \logo{\frac {w^{t+1}\of{\xiit}}{w^t_i}} + p^t\of{\xiit} \\
&= U^t\of{\xiit} + \frac {2\rho}{\delta} \logo{1 + \frac {\tau^t_i}{w^t_i}} + p^t\of{\xiit} \\
&> U^t\of{\xiit} + \frac {2\rho}{\delta} \frac {\tau^t_i}{w^t_i} - 2 \rho \delta + p^t\of{\xiit} \\
&= U^t\of{\xiit} + q^t_i - 2 \rho \delta
\end{align*}
where we have used the bound $\abs{\frac {\tau^t_i}{w^t_i}} \leq \delta$
and the inequality $\logo{1 + z} > z - \delta^2$ for $\abs{z} \leq \delta \leq \frac 12$.
\end{proof}

Our third lemma relates the quantities $U^T\of{x}$ to the payoffs $p^{\leq T}\of{x}$:
\begin{lemma}\label{exantereduction3}
\[ \PtH \geq \sum_{x \in H} U^T\of{x} - \frac {2 \rho}{\delta} \abs{H} \logo{\frac {N}{\abs{H}}} \]
\end{lemma}
\begin{proof}
$\sum_x w^t\of{x}$ is conserved by the algorithm, and $w^t\of{x} \geq 0$,
so 
\[ \sum_{x \in H} w^T\of{x} \leq \sum_x w^T\of{x} = \sum_x w^0\of{x} = N.\]
By Jensen's inequality, $\sum_x \logo{w^T\of{x}} \leq \abs{H} \logo{\frac {N}{\abs{H}}}$.
We have
\begin{align*}
\PtH 
&= \sum_{x \in H} U^T\of{x} - \frac {2\rho}{\delta} \sum_{x \in H} \logo{w^{T}\of{x}} \\
&\geq \sum_{x \in H} U^T\of{x} - \frac {2\rho}{\delta} \abs{H}\logo{\frac {N}{\abs{H}}},
\end{align*}
as desired.
\end{proof}

Combining these three lemmas yields our desired result:
\begin{theorem}\label{exantesymmetry}
There is an algorithm that obtains regret
\[ \PtH \geq \OPTH - \REGRET \]
assuming that the payoffs between honest users satisfy ex ante symmetry.
\end{theorem}
\begin{proof}
We will assume that $T$ is known in advance,
and choose $\delta = \sqrt{NT}$; this assumption
can then be removed by a standard doubling trick exactly as in
Theorem~\ref{thm:centralplanner}.

In any interaction between honest users we have $q_0^t = q_1^t$.
Moreover, $q_i^t \in [-\rho, \rho]$.
Thus we can apply $\symmetricalg$ and Lemma~\ref{exantereduction1}
and obtain in expectation
\begin{align*}
\sumsub{t, i : \xiit \in H}{s_0^t = s_1^t = 1} q_i^t
&\geq \sum_{t : \xot, \xit \in H}\of{ q_0^t + q_1^t} - \REGRET \\
&= \sum_{t : \xot, \xit \in H} \of{p_0^t + p_1^t} - \REGRET.
\end{align*}
Then we can apply Lemma~\ref{exantereduction2} and Lemma~\ref{exantereduction3}
and obtain in expectation
\begin{align*}
\PtH
&\geq \sum_{x \in H} U^T\of{x} - \frac {2\rho}{\delta} \abs{H} \logo{\frac {N}{\abs{H}}} \\
&\geq \sum_{x \in H} U^T\of{x} - \frac {2\rho N}{\delta}  \\
&= \sum_{x \in H} U^T\of{x} - \O{ \rho \sqrt{N T}} \\
&\geq \sum_{x \in H} U^0\of{x} + \sumsub{i, t : x_i^t \in H}{s_0^t = s_1^t = 1} q^t_i - \rho \delta T - \O{\rho \sqrt{N T }} \\
&= \sumsub{i, t : x_i^t \in H}{s_0^t = s_1^t = 1} q^t_i - \O{\rho \sqrt{N T }} \\
&\geq \sum_{t : \xot, \xit \in H} \of{p_0^t + p_1^t} - \REGRET - \O{\rho \sqrt{N T}} \\
&\geq \OPTH - \REGRET
\end{align*}
where the last inequality assumes $T > N$. If $T \leq N$ the result is trivial, since 
in this case the maximum possible difference between $\PtH$ and $\OPTH$ is $2 \rho T = \REGRET$ .
\end{proof}

\newcommand{\OPTHone}{\text{OPT}^{\leq T}\of{H_1}}
\newcommand{\OPTHtwo}{\text{OPT}^{\leq T}\of{H_2}}
\newcommand{\PtHone}{p^{\leq T}\of{H_1}}
\newcommand{\PtHtwo}{p^{\leq T}\of{H_2}}
\section{Appendix: Lower bounds}

In this section we consider the game described in Section~\ref{centralplanner}.
We show that for any $T$, there is a set $H$ and a strategy for nature such that
the planner's payoff is
\[ \OPTH - \Omega\of{\sqrt{NT}} \]
These results extend immediately to all of our other settings,
because an algorithm for any of these settings could be trivially
used by the central planner (who could simply simulate the interactions of the users
and recommend that two users interact with whatever probability they interacted
in the simulation).

To prove this result we consider a randomized strategy for nature
such that the expected payoff of every user is $0$.
We then consider the best possible partition of the users into two sets
$H_1, H_2$, and show that $\OPTHone + \OPTHtwo = \Omega\of{\sqrt{N T}}$.
This implies that for some choices by nature, and for one of the sets $H_i$,
the payoff of the central planner is less than $\text{OPT}^{\leq T}\of{H_i} - \Omega\of{\sqrt{TN}}$.

\begin{theorem}\label{lower-bound}
For every algorithm for the central planner and every $T$,
there exists a set of payoffs and a set $H$ such that
\[ \OPTH - o\of{\sqrt{NT}} \]
in expectation.
\end{theorem}
\begin{proof}
Consider a strategy for nature in which each $\xot$ and $\xit$ is uniformly random,
and each $p^t$ is a uniform $\pm 1$ random variable.
It is easy to see that no matter how the central planner plays in this game,
the expectation of $\PtH = 0$.

Given some enumeration of the users $0, 1, \ldots, N-1$,
we define the set $H_1$ as follows.
We include all of the users $0, 1, \ldots, N/2$ in $H_1$.
We then inductively decide whether to include $i > N/2$
in $H_1$ by evaluating the sum
\[ P_x = \sum_{\substack{t \leq T : \xit < \xot = x \\  \xit \in H_1}} p^t + \sum_{\substack{t \leq T : \xot < \xit = x \\  \xot \in H}} p^t.\]
We include $x$ in $H_1$ if and only if $P_x > 0$.

The sum defining each $P_x$ has in expectation at least $T/(2N)$ terms,
because $1/N$ of all interactions involve user $x$
and at least half of these interactions are with a user in $\abs{H_1}$.
Each of these terms is a uniform $\pm 1$ random variable.
So the probability that $P_x > 0$ is $\frac 12$ and the expectation of $P_x$, 
given that $P_x > 0$, is at least $\sqrt{T/N}$.

It is easy to see that $\OPTHone = \sum_{x \in H_1} P_x$.
Each term $P_x$ for $x \leq N/2$ has expectation 0.
In expectation there are $\Omega\of{N}$  terms with $x > N/2$.
Each of these terms is $\Omega\of{\sqrt{T/N}}$ in expectation.
Thus $\OPTHone = \Omega\of{\sqrt{TN}}$.

We define $H_2 = \users \backslash H_1$.
We have $\OPTHtwo = 0$:
for each pair $x < y$ with $x, y \in H_2$,
the event $y \in H_2$ depends only on the sum $P_y$,
which is a sum of terms independent of any interactions
occurring between $x$ and $y$.
Thus if there is an interaction between $x$ and $y$, its expected
payoff is $0$.
Summing over pairs of users in $H_2$, we obtain $\E{\text{OPT}\of{H_2}} = 0$.

Thus $\E{\OPTHone + \OPTHtwo} = \Omega\of{\sqrt{TN}}$, 
while $\E{\PtHone + \PtHtwo} = \E{p^{\leq T}\of{\users}} = 0$.
In particular, $\E{\OPTHone + \OPTHtwo - \PtHone - \PtHtwo} = \Omega\of{\sqrt{TN}}$,
and so we can find some particular strategy by nature and a set $H_i$
such that $p^{\leq T}\of{H_i} \leq \text{OPT}^{\leq T}\of{H_i} - \Omega\of{\sqrt{TN}}$, as desired.
\end{proof}

\section{Appendix: Strengthened analysis of online local learning}\label{ollsection}

\begin{lemma}\label{smoothness}\cite{me}
For any $E^{< t}$ and any $E^t$ we have
\[ \abs{\myalgo{E^{<t}} - \myalgo{E^{\leq t}}}_{\infty} \leq \O{\epsilon \abs{E^t}_1} \]
\end{lemma}
\begin{proof}
Note that the $\ell_1$ and $\ell_{\infty}$ norms are dual.
So this result follows directly from Lemma 8 in \cite{me}
and Lemma 2.19 in \cite{shai}.
\end{proof}

\begin{theorem}\cite{me}\label{oll}
For any sequence $E^t$ and any $X \succeq 0$ with entries in $[0, 1]$,
we have
\[ \sum_t \mprod{E^t}{\myalgo{E^{<t}}} \geq \sum_t \mprod{E^t}{X} - \O{\epsilon \sum_{t < T} \abs{E^t}_1^2 } - \O{N \epsilon^{-1}} \]
where $\abs{E^t}_1$ is the sum of the absolute values of the entries of $E^t$.
\end{theorem}
\begin{proof}
By Lemma 2.3 of \cite{shai} and Lemma 5 of \cite{me}, for any $X \succeq 0$ with entries
in $[0, 1]$ we have:
\[ \sum_t \mprod{E^t}{\myalgo{E^{<t}}} \geq \sum_t \mprod{E^t}{X} - \O{N \epsilon^{-1}}
+ \sum_t \mprod{E^t}{\myalgo{E^{<t}} - \myalgo{E^{\leq t}}}. \]
Using Lemma~\ref{smoothness} we have
\begin{align*}
\sum_t \mprod{E^t}{\myalgo{E^{<t}} - \myalgo{E^{\leq t}}}
&\geq - \sum_t \abs{E^t}_1 \abs{\myalgo{E^{<t}} - \myalgo{E^{\leq t}}}_{\infty} \\
&= \O{\sum_t \epsilon \abs{E^t}_1^2}
\end{align*}
which gives the desired result.
\end{proof}

\section{Appendix: Open questions}\label{further-work}

We have presented a protocol for managing trust
which achieves a qualitatively new---and surprisingly strong---formal guarantee.
But the theoretical picture is far from complete,
and the proposed protocol is far from practical.
We hope that our model and algorithmic approach
will inspire further work in both of these directions,
and in this section we highlight some areas for improvement and elaboration.

\begin{description}
\item[Tighter regret bounds:] Our regret bound of $\O{\rho \sqrt[3]{T^2 N}}$ can probably
be improved to $\O{\rho \sqrt{T N}}$, which translates into a significant
performance improvement.

Perhaps more importantly, our regret bound depends on the total number of interactions.
This might be improved to depend only on the number of interactions involving an honest user.
The distinction is particularly important in settings where dishonest users
can fabricate large numbers of interactions amongst themselves.

With plausible parameters and performance goals in the favor setting, 
these improvements could reduce the number
of required interactions per user from tens of thousands to hundreds\footnote.{although
the practical performance of the protocol might be better than the theoretical worst case,
as often happens}
Other performance improvements, such as distinguishing between the average
and maximal payoffs or
improving the hidden constant factors,
would also help achieve a practical guarantee.

\item[Distributed implementation:] Although our protocol is is 
tractable in the sense of polynomial-time,
it is not distributed.
Ideally each user's computation 
and communication would be polylogarithmic in the size of the community.

The main challenge is distributing the computation of the central planner,
and particularly the online learning algorithm at the core of our protocol.
The current algorithm requires semidefinite programming, 
which must be replaced
by fast distributed computation, such as random walks
or sparse linear algebra.
In addition to the algorithmic difficulties,
the computation must be robust to manipulation by adversarial users;
this might be achieved either by finding a computation which is inherently robust
or by using cryptographic tools to ensure accurate computation.

\item[Simultaneous interactions:] Our protocol assumes that interactions occur one at a time:
after one pair of users interacts, their reports immediately become available
to the next pair of users.
In most contexts this is not a realistic assumption.

When applied to a setting with $k$ simultaneous interactions,
our protocol achieves regret $\Omega\of{k}$.
It is not hard to construct examples
where the regret must necessarily be $\Omega\of{\log{k}}$.
We suspect that this lower bound is tight, 
and that some protocol achieves regret $\O{\rho \log{k} \sqrt{T N}}$.
Under realistic assumptions, it may be possible to eliminate the dependence on $k$
altogether.

\item[Recommending vs. filtering:] Our protocol answers questions of the form ``should I interact with this merchant?''
rather than questions of the form ``which merchant should I interact with?''
In many contexts the latter question is more useful---though deciding which users' input to trust
is a key difficulty for both problems.
\cite{ccl} provides a similar algorithm that answers ``which merchant should I interact with?''
but exploits the simplifying assumption that there is a single best merchant for all transactions.
It is natural to try to combine these two assumptions,
and to design a protocol which recommends one resource from amongst $k$ options.

\item[Strong signals:] In our setting, most interactions provide only weak
evidence about which users are honest.
We might expect to be able to obtain tighter bounds when very strong signals are available.
For example, a dishonest user might physically assault
or threaten an honest user, resulting in a payoff of $-\rho^*$
which could never occur in an interaction between honest users.
In this case, 
we would like to obtain a regret bound whose dependence
on $\rho^*$ was $\O{\rho^* N}$, independent of $T$.

\item[Incentive compatibility:] Our protocol has a robust theoretical guarantee,
but only assuming that there are some honest users who reliably follow the protocol.
It would be better if honesty were incentivized rather than assumed.

The honest players in our protocols report their payoff after each interaction,
and pay a cost which depends on their reported payoff.
So an honest player could always benefit by stating a lower payoff,
and has no short-term incentive to play by the rules.

It is worth noting, however, that a player who follows this selfish approach
will quickly find that no one wants to interact with them.
Though the optimal behavior involves lying,
it involves lying just enough that other users still benefit from interacting with me.
In fact, 
a community of users lying optimally can still lead to socially optimal outcomes.

Unfortunately, in an attempt to lie optimally an honest user may inadvertently destroy value.
In fact the Myerson-Satterthwaite theorem applies to our setting,
even when there are just two users,\footnote{Here is a very rough sketch of the argument.
Consider our setting with transfers,
where
$-p_0^t, p_1^t$ are independently distributed in $[0, 1]$
and are the same in every round.
Because $H$ might contain only one user or might contain both,
a protocol that satisfies our guarantee must be incentive compatible
and must maximize collective welfare (i.e. be ex post efficient).
We can use the average size of the transfer
to sell an item valued at $p_0^t$ by the seller and $p_1^t$ by the buyer,
in a way that is incentive compatible and ex post efficient.
By Myerson-Satterthwaite, we can conclude that any protocol
satisfying our guarantees is not a Bayesian Nash equilibrium.

In the setting of bilateral exchange, our protocol recommends that the honest users
honestly state their willingness to pay for the item, allocate it optimally,
and then divide the surplus evenly. Evidently this is not a Nash equilibrium,
since a player can benefit by understating their surplus. This case captures
the essential and unavoidable reason that our protocol is not incentive-compatible.} to show 
that this is necessarily the case:
in general it is impossible to achieve our desired guarantees at a Bayesian Nash equilibrium.
We feel that our result is an attractive compromise given this impossibility result.

We could still achieve actual incentive compatibility
or approximate incentive compatibility in many interesting cases
where fundamental obstructions don't arise;
in general, we might hope to translate positive results
from the two-player iterated setting to our distributed setting.

\item[Asymmetric interactions:] Our protocols all require that interactions be symmetric
\emph{on average}, or else make use of transfers to symmetrize them.
This is a very strong condition that would be nice to weaken.

One generalization is to consider sets of interactions that are balanced,
so that each user has as many opportunities to give as to receive overall,
even though the pairwise interactions are asymmetrical.

A more difficult generalization is to consider interactions that are inherently imbalanced,
in the sense that some users systematically have more opportunities to give
than to receive.
In this case it is no longer possible to simultaneously maximize collective welfare
for every subgroup $H$.
We could instead try to achieve Pareto efficiency,
which is a weaker condition but still meaningful
if it holds for \emph{every} set $H$ simultaneously.
For example, we might
try to ensure that no group $H$ can benefit all of its members 
by adopting a different internal policy
and/or isolating itself from the users outside of $H$.

\item[Honest supermajorities:] Our guarantees are non-trivial even when
the set $H$ is a very large fraction of $\users$.
But in this case it may be able to obtain regret bounds that depend on $\of{N - \abs{H}}$
rather than on $N$.
In cases where 90\% or 99\% of users are honest,
this could amount to a large difference in performance.

\item[Sybilproofness:] Our system bounds the damage done by the introduction
of any dishonest user.
In some settings it is easy for an attacker to create large numbers of dishonest users
or \emph{sybils}.
Our bounds are not tight enough to discourage such an attack.

It might be possible to ensure that the introduction
of a dishonest user does \emph{no} damage to the honest users,
or that the maximum possible damage done by dishonest users is bounded.
For example, if transfers are available, each user could pay an upfront
deposit to cover any possible damage they might inflict (in particular,
a new user would need to pay a cost of at least $\rho$ before \emph{every} interaction
with an established user,
which would potentially be refunded if the interaction went well).
Each honest user would need to pay a manageable cost to join the system,
but the cost would be large enough that attackers can only recoup it by
becoming productive members of the community.\footnote{Impossibility results 
such as \cite{sybilproofaccounting} do apply to this problem,
but we are not aware of any that apply to our approach.}

\item[Composition and complex protocols:] Worst-case performance guarantees
allow our protocol to be used as a subroutine
without compromising its correctness.
Together with other building blocks,
our result might be used to design more complex protocols
in the same way that simple cryptographic primitives are used to build
more complex cryptographic systems.

\end{description}

%XXX be a bit more careful about the reductions

\end{document}